\documentclass[11pt]{article}

\usepackage[utf8]{inputenc}
\usepackage[margin=1in]{geometry}
\usepackage[numbers]{natbib}
\usepackage{floatpag}
\usepackage{float}
\usepackage{wrapfig}
\newfloat{algorithm}{t}{lop}

\usepackage{amsmath}
\usepackage{amssymb}
\usepackage{amsthm}
\usepackage{amsfonts}
\usepackage{mathtools}
\usepackage{bbm}
\usepackage{pifont}
\usepackage{dsfont}
\usepackage{microtype}
\usepackage{booktabs}
\usepackage{subcaption}
\usepackage{graphicx}
\usepackage[dvipsnames]{xcolor}
\usepackage{algorithm}
\usepackage[noend]{algorithmic}

\newcommand{\alglinelabel}{%
  \addtocounter{ALC@line}{-1}
  \refstepcounter{ALC@line}
  \label
}

\usepackage{thmtools}
\usepackage{thm-restate}

\usepackage{hyperref}

\usepackage[capitalize,noabbrev]{cleveref}
\creflabelformat{equation}{#2#1#3}

\newcommand{\eps}{\ensuremath{\varepsilon}}

\newcommand{\calA}{\ensuremath{\mathcal{A}}}

\newcommand{\calX}{\ensuremath{\mathcal{X}}}

\theoremstyle{plain}
\newtheorem{theorem}{Theorem}[section]

\newtheorem{lemma}[theorem]{Lemma}

\theoremstyle{definition}
\newtheorem{definition}[theorem]{Definition}

\theoremstyle{remark}

\newcommand{\tm}{\tilde m}
\newcommand{\tu}{\tilde U}

\DeclareMathOperator*{\argmax}{arg\,max}
\DeclareMathOperator*{\argmin}{arg\,min}
\newcommand{\indic}[1]{\mathbbm{1}_{#1}}
\renewcommand{\P}[2]{\mathbb{P}_{#1}\left[#2\right]}

\newcommand{\ellk}{\ell_{\textsf{k-rel}}}
\newcommand{\topk}{\textsc{top-k}}

\newcommand{\emech}{\textsc{EM}}
\newcommand{\joint}{\textsc{Joint}}
\newcommand{\cpeel}{\textsc{CDP-Peel}}
\newcommand{\ppeel}{\textsc{PNF-Peel}}
\newcommand{\gum}[1]{\textsc{Gumbel}\left(#1\right)}

\newcommand{\expo}[1]{\textsc{Expo}\left(#1\right)}
\newcommand{\rnmexpo}{\textsc{RnmExpo}}
\newcommand{\pnfjoint}{\textsc{PNFJoint}}
\newcommand{\maxexpo}{\textsc{MaxExpo}}

\renewcommand{\epsilon}{\varepsilon}

\definecolor{emerald}{RGB}{0,153,123}

\hypersetup{
    colorlinks=true,
    linkcolor=blue,
    urlcolor=blue,
    citecolor=blue
}

\newcommand{\arxiv}[1]{#1}
\newcommand{\narxiv}[1]{}

\renewenvironment{abstract}
 {\small
  \begin{center}
  \bfseries \abstractname\vspace{-.5em}\vspace{0pt}
  \end{center}
  \list{}{
    \setlength{\leftmargin}{6mm} 
    \setlength{\rightmargin}{\leftmargin}
  }
  \item\relax}
 {\endlist}

\begin{document}

\title{A Joint Exponential Mechanism For Differentially Private Top-$k$}
\author{Jennifer Gillenwater\thanks{Google New York, jengi@google.com}\and Matthew Joseph\thanks{Google New York, mtjoseph@google.com} \and Andrés Muñoz Medina\thanks{Google New York, ammedina@google.com} \and Mónica Ribero\thanks{UT Austin. Part of this work done while an intern at Google New York.}}
\maketitle

\begin{abstract}
    We present a differentially private algorithm for releasing the sequence of $k$ elements with the highest counts from a data domain of $d$ elements. The algorithm is a ``joint'' instance of the exponential mechanism, and its output space consists of all $O(d^k)$ length-$k$ sequences. Our main contribution is a method to sample this exponential mechanism in time $O(dk\log(k) + d\log(d))$ and space $O(dk)$. Experiments show that this approach outperforms existing pure differential privacy methods and improves upon even approximate differential privacy methods for moderate $k$.
\end{abstract}
\section{Introduction}
Top-$k$ is the problem of identifying the ordered sequence of $k$ items with the highest counts from a data domain of $d$ items. This basic problem arises in machine learning tasks such as recommender systems, basket market analysis, and language learning. To solve these problems while guaranteeing privacy to the individuals contributing data, several works have studied top-$k$ under the additional constraint of differential privacy (DP)~\cite{DMNS06}. Differential privacy guarantees that publishing the $k$ identified elements reveals only a controlled amount of information about the users who contributed data to the item counts.

The best known DP algorithm for top-$k$ is the ``peeling'' mechanism~\cite{BLST10, DR19}, which applies a DP subroutine for selecting the highest count item from a set, removes it, and repeats $k$ times. One possible such subroutine is the exponential mechanism, a general DP algorithm for choosing high-utility (here, high-count) elements from a data universe given some utility function (see \cref{sec:prelims}). Another similar subroutine is the permute-and-flip mechanism~\cite{MS20}. This leads to the baseline peeling mechanisms in our experiments, each using the best known composition methods: the approximate DP baseline uses the exponential mechanism analyzed via concentrated differential privacy (CDP) composition~\cite{DR16, BS16}, and the pure DP basline uses the permute-and-flip mechanism with basic composition. The approximate DP variant takes time $O(d + k\log(k))$ and the pure variant takes time $O(dk)$.  Both require space $O(d)$.

\subsection{Our Contributions}
\label{subsec:contributions}
We construct an instance of the exponential mechanism that chooses directly from sequences of $k$ items. Unlike the peeling mechanism, this approach does not use composition.  Past work has used this style of ``joint'' exponential mechanism to privately and efficiently estimate: 1-way marginals under $\ell_\infty$ error~\cite{SU15}, password frequency lists~\cite{BDB16}, and quantiles~\cite{GJK21}. However, it is not obvious how to extend any of these to top-$k$ selection.

Naive implementation of a joint exponential mechanism, whose output space is all sequences of $k$ items, requires enumerating all $O(d^k)$ such sequences.  This is impractical even for modest values of $d$ and $k$. As with previous work on joint exponential mechanisms, our main contribution is an equivalent efficient sampling method.

\begin{theorem}[Informal version of \cref{thm:main}]
    There is a joint exponential mechanism for $\eps$-DP top-$k$ that takes time $O(dk\log(k) + d\log(d))$ and space $O(dk)$.
\end{theorem}

While it is straightforward to prove a utility guarantee for this mechanism (\cref{thm:main_utility}), our main argument for this joint approach is empirical, as asymptotic guarantees often obscure markedly different performance in practice. Experiments show that the joint exponential mechanism offers the strongest performance among pure differential privacy mechanisms\arxiv{\footnote{It is also possible, with a small change to the proposed joint mechanism, to switch it from an instance of the exponential mechanism to an instance of the permute-and-flip mechanism~\cite{MS20}. This can theoretically improve utility by up to a factor of $2$.  However, the change appears to be negligible for typical datasets and raises numerical issues beyond small $k$.  See \cref{sec:pnf} for details.}} and even outperforms approximate differential privacy mechanisms when $k$ is not large (\cref{sec:experiments}).

\subsection{Related Work}
\label{subsec:related}
Private top-$k$ was first studied in the context of frequent itemset mining, where each user contributes a set of items, and the goal is to find common subsets~\cite{BLST10, LQSC12, ZNC12, LC14}. \citet{BLST10} introduced the first version of the peeling mechanism. Our main points of comparison will be variants of the peeling mechanism developed by \citet{DR19} and \citet{MS20}.

Past work has also studied several other algorithms for DP top-$k$. Laplace noise has been used for pure and approximate DP~\cite{DWZK19, QSZ21}. Additionally, for pure DP, the Gamma mechanism for releasing private counts (Theorem 4.1, \cite{SU15}) can be applied. Our experiments found that peeling mechanisms dominate these approaches, so we omit them, but implementations of all three appear in \narxiv{the supplementary material}\arxiv{our public code respository~\cite{Go22}}.  Finally, we note that \citet{DR19} study the problem of private top-$k$ when the number of items $d$ is prohibitively large and $O(d)$ runtime is impractical. We instead focus on the setting where $O(d)$ runtime is acceptable.

A few lower bounds are relevant to private top-$k$. One naive approach is to simply privately estimate all $d$ item counts and return the top $k$ from those noisy estimates. However, each count has expected error $\Omega(d/\eps)$ for pure DP (Theorem 1.1, \cite{HT10}). Similarly,~\citet{BUV14} (Corollary 3.4) construct a distribution over databases such that, over the randomness of the database and the mechanism, each count has expected error $\Omega(\sqrt{d}/\eps)$ for approximate DP. Top-$k$ approaches aim to replace this dependence on $d$ with a dependence on $k$.~\citet{BU17} and~\citet{SU17} prove lower bounds for what we call $k$-relative error, the maximum amount by which the true $k^{th}$ count exceeds any of the estimated top-$k$ counts (see \cref{sec:prelims} for details). Respectively, they prove $\Omega(k\log(d))$ and $\Omega(\sqrt{k}\log(d))$ sample complexity lower bounds for ``small'' and ``large'' $k$-relative error. In both cases, the peeling mechanism provides a tight upper bound. We upper bound signed maximum error (\cref{thm:main_utility}), but our paper more generally departs from these works by focusing on empirical performance.
\section{Preliminaries}
\label{sec:prelims}
\paragraph{Notation.} $[m] = \{1, \ldots, m \}$. Given a vector $v$ in dimension $d$ and indices $i<j$, $v_{i:j}$ denotes coordinates $i, \ldots, j$ of $v$; given a sequence $S$, $v_S$ denotes coordinates $v_i$ for $i \in S$. 

We start by formally describing the top-$k$ problem.

\begin{definition}
\label{def:top_k}
    Let $\calX$ be a data domain of $d$ items. In an instance of the \emph{top-$k$} problem, there are $n$ users and each user $i \in [n]$ has an associated vector $x_i \in \{0,1\}^d$. For dataset $D = \{x_i\}_{i=1}^n$, let $c_j = \sum_{i=1}^n x_{i,j}$ denote the count of item $j \in [d]$, and let $(c_1, \ldots, c_d)$ denote the counts in nonincreasing order. Given sequence of indices $S$, let $c_S$ denote the corresponding sequence of counts. Given sequence loss function $\ell$, the goal is to output a sequence $S$ of $k$ items $\topk(D) = \argmin_{S = (s_1, \ldots, s_k)} \ell(c_S)$.
\end{definition}

Note that each user contributes at most one to each item count but may contribute to arbitrarily many items. This captures many natural settings. For example, a user is unlikely to review the same movie more than once, but they are likely to review multiple movies. In general, we describe dataset $D$ by the vector of counts for its domain, $D = (c_1, \ldots, c_d)$.

Our experiments will use $\ell_\infty$, $\ell_1$, and, in keeping with past work on private top-$k$~\cite{BU17, DR19}, what we call $k$-relative error.

\begin{definition}
    \label{def:errors}
    Using the notation from \cref{def:top_k}, we consider sequence error functions
    \begin{enumerate}
    \setlength\itemsep{0.05em}
        \item $\ell_\infty(c_S) = \|c_{1:k} - c_S\|_\infty$,
        \item $\ell_1(c_S) = \|c_{1:k} - c_S\|_1$,  and
        \item \emph{$k$-relative error} $\ellk(c_S) = \max_{i \in [k]}(c_k - c_{s_i})$.
    \end{enumerate}
\end{definition}

The specific choice of error may be tailored to the data analyst's goals: $\ell_\infty$ error suits an analyst who wishes to minimize the worst error of any of the top $k$ counts; $\ell_1$ error is appropriate for an analyst who views a sequence of slightly inaccurate counts as equivalent to one highly inaccurate count; and $k$-relative error may be best when the analyst prioritizes a ``sound'' sequence where no count is much lower than the true $k^{th}$ count. Note that, while $k$-relative error has been featured in past theoretical results on private top-$k$ selection, it is the most lenient error metric. For example, given $d$ items with counts $100, 1, \ldots, 1$ and $k=2$, \emph{any} sequence of items obtains optimal $k$-relative error, and in general $\ellk(c_S) \leq \min(\ell_\infty(c_S), \ell_1(c_S))$.

Next, we cover privacy prerequisites. Differential privacy guarantees that adding or removing a single input data point can only change an algorithm's output distribution by a carefully controlled amount.

\begin{definition}[\citet{DMNS06}]
         Datasets $D, D' \in \calX^*$ are \emph{neighbors} (denoted $D \sim D'$) if $D'$ can be obtained from $D$ by adding or removing a data point $x$. Mechanism $M \colon \calX^* \to Y$ is \emph{$(\eps,\delta)$-differentially private} (DP) if, for any two neighboring datasets $D \sim D'$ in $\calX^*$, and any $S\subseteq Y$, it holds that $\P{}{M(D) \in S} \leq e^\epsilon \P{}{M(D') \in S} + \delta$. If $\delta=0$, it is $\eps$-DP.
\end{definition}

One especially flexible differentially private algorithm is the exponential mechanism. Given some utility function over outputs, the exponential mechanism samples high-utility outputs with higher probability than low-utility outputs.

\begin{definition}[\citet{MT07, DR14}]
\label{def:exp}
    Given utility function $u \colon \calX^* \times O \to \mathbb{R}$ with $\ell_1$ sensitivity $\Delta(u) = \max_{D \sim D', o \in O} |u(D, o) - u(D', o)|$, the exponential mechanism $M$ has output distribution
    \[
        \P{}{M(D) = o} \propto \exp\left(\frac{\eps u(D,o)}{2\Delta(u)}\right),
    \]
    where $\propto$ elides the normalization factor.
\end{definition}

Note that this distribution places relatively more mass on outputs with higher scores when the sensitivity $\Delta(u)$ is small and the privacy parameter $\eps$ is large.

\begin{lemma}[\citet{MT07}]
\label{lem:em}
    The exponential mechanism is $\eps$-DP.
\end{lemma}

A tighter analysis is possible for certain utility functions.

\begin{definition}
\label{def:monotonic}
    A  utility function $u$ is \emph{monotonic} (in the dataset) if, for every dataset $D$ and output $o$, for any neighboring datasets $D'$ that results from adding some data point to $D$, $u(D, o) \leq u(D', o)$.
\end{definition}

When the utility function is monotonic, the factor of 2 in the exponential mechanism's output distribution can be removed. This is because the factor of 2 is only necessary when scores for different outputs move in opposite directions between neighboring datasets.

\begin{lemma}[\citet{MT07}]
    Given monotonic utility function $u$, the exponential mechanism is $\tfrac{\eps}{2}$-DP.
\end{lemma}

One cost of the exponential mechanism's generality is that its definition provides no guidance for efficiently executing the sampling step. As subsequent sections demonstrate, this is sometimes the main technical hurdle to applying it.
\section{A Joint Exponential Mechanism for Top-$k$}
\label{sec:joint}
Our application of the exponential mechanism employs a utility function $u^*$ measuring the largest difference in counts between the true counts and candidate counts. Let $(c_1, \ldots, c_d)$ be the item counts in nonincreasing order. For candidate sequence of items $S = (s_1, \ldots, s_k)$, we define
\[
  u^*(D, S) =
  \begin{cases}
        -\max_{i \in [k]}(c_i - c_{s_i}) & \text{if $s_1, \ldots, s_k$} \\
                & \text{are distinct.}\\
        -\infty & \text{otherwise}
  \end{cases}
\]
$u^*$ thus assigns the highest possible score, 0, to the true sequence of top $k$ counts, and increasingly negative scores to sequences with smaller counts. Sequences with repeated items have score $-\infty$ and are never output.

\paragraph{Discussion of $u^*$.} A natural alternative to $u^*$ would replace $-\max_{i \in [k]}(c_i - c_{s_i})$ with $-\max_{i \in [k]}|c_i - c_{s_i}| = -\|c_i - c_{s_i}\|_\infty$. Call this alternative $u'$. In addition to being expressible as a simple norm, $u'$ also corresponds exactly to the number of user additions or removals sufficient to make $S$ the true top-$k$ sequence\footnote{The idea of a utility function based on dataset distances has appeared in the DP literature under several names~\cite{JS13, AD20a, MG20} but has not been applied to top-$k$ selection.}. However, $u^*$ has two key advantages over $u'$. First, $u^*$ admits an efficient sampling mechanism. Second, $u'$ favors sequences that omit high-count items entirely over sequences that include them in the wrong order. For example, suppose we have a dataset $D$ consisting of $d = 10$ items with counts $100, 90, \ldots, 10$.  If we want the top $k=5$ items, we will consider sequences such as $S_1 = (1, 3, 4, 5, 2)$ and $S_2 = (1, 3, 4, 5, 6)$.  These have identical value according to $u^*$: $u^*(D,S_1) = -10 = u^*(D,S_2)$.  But according to $u'$, $S_1$ scores much worse than $S_2$: $u'(D,S_1) = -30 < -10 = u'(D,S_2)$. This conflicts with the ultimate goal of identifying the highest-count items; $S_1$ contains item 2 (count $90$), while $S_2$ replaces it with item 6 (count $50$)\footnote{The standard $\ell_{\infty}$ loss metric shares this flaw; $\ell(c_S) = \max_{i \in [k]}(c_i - c_{s_i})$ may therefore be a reasonable loss metric for future top-$k$ work. Nonetheless, past work uses $\ell_\infty$ error, and we did not observe large differences between the two empirically, so we use $\ell_\infty$ in our experiments as well.}. We now show that $u^*$ also has low sensitivity.

\begin{lemma}
\label{lem:sensitivity}
    $\Delta(u^*) = 1$.
\end{lemma}
\begin{proof}
    First, any sequence with utility $-\infty$ has that utility on every dataset. Turning to sequences of distinct elements $(s_1, \ldots, s_k)$, adding a user does not decrease any count, and increases a count by at most one. Furthermore, while the top $k$ items may change, none of the top-$k$ counts decrease, and each increases by at most one. It follows that each $c_i - c_{s_i}$ either stays the same, decreases by one, or increases by one. A similar analysis holds when a user is removed.
\end{proof}

We call the instance of the exponential mechanism with utility $u^*$ \joint. Its privacy is immediate from \cref{lem:em}.

\begin{theorem}
\label{thm:main_privacy}
    \joint{} is $\eps$-DP.
\end{theorem}

A utility guarantee for \joint{} is also immediate from the generic utility guarantee for the exponential mechanism. The (short) proof appears in \cref{sec:utility-proof}.

\begin{restatable}{theorem}{MainUtility}
\label{thm:main_utility}
    Let $c_1, \ldots, c_k$ denote the true top-$k$ counts for dataset $D$, and let $\tilde c_1, \ldots, \tilde c_k$ denote those output by \joint. With probability at least $99/100$, 
    \[
        \max_{i \in [k]}(c_i - \tilde c_i) \leq \frac{2[k\ln(d) + 5]}{\eps}.
    \]
\end{restatable}

Naive sampling of \joint{} requires computing $O(d^k)$ output probabilities. The next subsection describes a sampling algorithm that only takes time $O(dk\log(k) + d\log(d))$.

\subsection{Efficiently Sampling \joint}
\label{subsec:sampling}
The key observation is that, while there are $O(d^k)$ possible output sequences, a given instance $u^*(D,\cdot)$ has only $dk$ possible values. This is because each score takes the form $-(c_i - c_j)$ for some $i \in [k]$ and $j \in [d]$. Our algorithm will therefore proceed as follows:

\begin{enumerate}
    \item For each of the $O(dk)$ utilities $U_{ij} = -(c_i - c_j)$, \textbf{count the number} $m(U_{ij})$ of sequences $S$ with score $u^*(S) = U_{ij}$. 
    \item \textbf{Sample a utility} $U_{ij}$ from the distribution defined by
    \begin{equation}
    \label{eq:utility_distribution}
        \P{}{U_{ij}} \propto m(U_{ij})\exp\left(\frac{\eps U_{ij}}{2}\right).
    \end{equation}
    \item From the space of all sequences that have the selected utility $U_{ij}$, \textbf{sample a sequence} uniformly at random.
\end{enumerate}

This outline makes one oversimplification: instead of counting the number of sequences for each of $O(dk)$ (possibly non-distinct) integral utility values, the actual sampling algorithm will instead work with exactly $dk$ distinct non-integral utility values. Nonetheless, the output distribution will be exactly that of the exponential mechanism described at the beginning of the section.

\subsubsection{Counting the Number of Sequences}
Define $k \times d$ matrix $\tu$ by $\tu_{ij} = -(c_i - c_j) - z_{ij}$ where $z_{ij}$ is a small term in $(0,1/2]$ that ensures distinctness, 
\[
    z_{ij} = \frac{d(k-i) + j}{2dk}.
\]
Several useful properties of $\tu$ are stated in \cref{lem:U}.

\begin{lemma}
\label{lem:U}
Given $\tu$ defined above, 1) each row of $\tu$ is decreasing, 2) each column of $\tu$ is increasing, and 3) the elements of $\tu$ are distinct.
\end{lemma}
\begin{proof}
    Fix some row $i$. By definition, $c_1 \geq c_2 \geq \cdots \geq c_d$, so $-(c_i - c_1) \geq \cdots \geq -(c_i - c_d)$. The $z_{ij}$ terms also increase with $j$, so each row of $\tu$ is decreasing. By similar logic, each column of $\tu$ is increasing.
    
    Finally, note that since $i \in [k]$ and $j \in [d]$, the $z$ terms are
    \[
        \frac{dk}{2dk}, \frac{dk-1}{2dk}, \ldots, \frac{2}{2dk}, \frac{1}{2dk}
    \]
    and thus are $dk$ distinct values in $(0, 1/2]$. Since any two count differences $-(c_{i_1} - c_{j_1})$ and $-(c_{i_2} - c_{j_2})$ are either identical or at least 1 apart, claim 3) follows.
\end{proof}

We now count ``sequences through $\tu$''. Each sequence $(s_1, \ldots, s_k)$ consists of $k$ values from $[d]$, one from each row of $\tu$, and its score is $\min_{i \in [k]} \tu_{is_i}$, or $-\infty$ if the $k$ values are not distinct. For each $i \in [k]$ and $j \in [d]$, define $\tm(\tu_{ij})$ to be the number of sequences through $\tu$ with distinct elements and score exactly $\tu_{ij}$. $\tm$ and $\tu$ are useful because of the following connection to $m(U_{ij})$, the quantities necessary to sample from the distribution in \cref{eq:utility_distribution}:

\begin{lemma}
    \label{lem:ms}
    For any $i \in [k]$ and $j \in [d]$, let $A_{ij} = \{\tu_{i'j'} \mid  \lceil \tu_{i'j'} \rceil = U_{ij}\}$. Then $m(U_{ij}) = \sum_{\tu_{i'j'} \in A_{ij}} \tm(\tu_{i'j'})$.
\end{lemma}
\begin{proof}
    Each $z \in (0,1/2]$, so $A_{ij}$ is exactly the collection of $\tu_{i'j'} = -(c_{i'} - c_{j'}) - z_{i'j'}$ where $c_{i'} - c_{j'} = c_i - c_j$. 
\end{proof}

The problem thus reduces to computing the $\tm$ values. For each row $r \in [k]$ and utility $\tu_{ij} \in \tu$, define
\[
    t_r(\tu_{ij}) = \max(\{j' \in [d] \mid \tu_{rj'} \geq \tu_{ij}\}).
\]

Useful simple properties of these $t_r$ values appear below.
\begin{lemma}
\label{lem:t_properties}
    Fix some $\tu_{ij} \in \tu$. Then 1) $t_r(\tu_{ij})$ is nondecreasing in $r$, 2) if sequence $S = (s_1, \ldots, s_k)$ has score $\tu_{ij}$, then for all $r \in [k]$, $s_r \leq t_r(\tu_{ij})$
    and 3) there exists a sequence $S = (s_1, \ldots, s_k)$ of distinct elements with score $\tu_{ij}$ if and only if $t_r(\tu_{ij}) \geq r$ for all $r$.
\end{lemma}
\begin{proof}
    The first two properties follow directly from \cref{lem:U}. For property 3) let $S = (s_1, \ldots, s_k)$ satisfy the conditions of the lemma and assume $t_r(\tu_{ij}) < r$ for some $r$. By properties 1) and 2), for all $r' \leq r$, $s_{r'} \leq t_{r'}(\tu_{ij}) \leq t_r(\tu_{ij})< r$. But that implies $S$ contains $r$ distinct numbers less than $r$, which is a contradiction. In the other direction, suppose $t_r(\tu_{ij}) \geq r$ for all $r$. Define $S = (s_1, \ldots, s_k)$ where $s_r = r$ for $r < i$ and $s_i = j$. Since $t_i(\tu_{ij}) \geq i$, $j \geq i$. Therefore $[t_{i+1}(\tu_{ij})] - \{s_1, \ldots, s_i\}$ contains at least one option for $s_{i+1}$, $[t_{i+2}(\tu_{ij})] - \{s_1, \ldots, s_{i+1}\}$ contains at least one option for $s_{i+2}$, and so on. The resulting $S$ has distinct elements and score $\tu_{ij}$.
\end{proof}
The following lemma connects the $t_r$ and $\tm$ values.
\begin{lemma}
\label{lem:counting_m}
    Given entry $\tu_{ij} $ of $ \tu$, define $n_r = \max(t_r(\tu_{ij}) - (r-1), 0)$. Then $\tm(\tu_{ij}) = \prod_{r \neq i} n_r$.
\end{lemma}
\begin{proof}
    By Lemma~\ref{lem:t_properties} we know the statement is true for $\tu_{ij}$ such that $\tm(\tu_{ij}) = 0$.
    If $\prod_{r \neq i} n_r > 0$ then any sequence with score $\tu_{ij}$ consists of distinct elements $S = (s_1, \ldots, s_k)$ such that $s_i = j$ and for all $r \neq i$, $s_r \in [t_r(\tu_{ij})]$ (otherwise $S$ has score  less than $\tu_{ij}$). The number of such sequences is $\prod_{r \neq i} [t_r(\tu_{ij}) - (r-1)]$.
\end{proof}

A naive solution thus computes all $dk^2$ of the $t_r$ values, then uses them to compute the $\tilde m$ values. We can avoid this by observing that, if we sort the values of $\tu$, then adjacent values in the sorted order have almost identical $t_r$.

\begin{lemma}
\label{lem:compute_ts}
    Let $\tu_{(1)}, \ldots, \tu_{(dk)}$ denote the entries of $\tu$ sorted in decreasing order. For each $\tu_{(a)}$, let $r(a)$ denote its row index in $\tu$. Then: 1) for each $a \in [dk-1]$, $t_{r(a+1)}(\tu_{(a+1)}) = t_{r(a+1)}(\tu_{(a)}) + 1$, and 2) for $r' \neq r(a+1)$, $t_{r'}(\tu_{(a+1)}) = t_{r'}(\tu_{(a)})$.
\end{lemma}
\begin{proof}
    For 1), assume $t_{r(a+1)}(\tu_{(a+1)}) > t_{r(a+1)}(\tu_{(a)}) + 1$. Then we can define $j$ such that \narxiv{$t_{r(a+1)}(\tu_{(a)}) < j < t_{r(a+1)}(\tu_{(a+1)})$,}\arxiv{\[t_{r(a+1)}(\tu_{(a)}) < j < t_{r(a+1)}(\tu_{(a+1)}),\]}
    which implies that $\tu_{(a)} > \tu_{r(a+1), j} > \tu_{(a+1)}$. This contradicts $\tu_{(a)}$ and $\tu_{(a+1)}$ being adjacent in sorted order. For 2), assume there exists $r' \neq r(a+1)$ such that $t_{r'}(\tu_{(a+1)}) > t_{r'}(\tu_{(a)})$. (If we instead assumed $t_{r'}(\tu_{(a+1)}) < t_{r'}(\tu_{(a)})$, this would contradict the sorting order of $\tu$ and the definition of $t_r$.) This implies that $\tu_{(a)} > \tu_{r',t_{r'}(\tu_{(a+1)})} > \tu_{(a+1)}$, which again contradicts $\tu_{(a)}$ and $\tu_{(a+1)}$ being adjacent in sorted order.
\end{proof}

According to the above lemma, we can compute all of the $\tm$ as follows.  First, sort the entries of $\tu$, recording the row and column indices of each $\tu_{(a)}$ as $r(a)$ and $c(a)$.  Then, create a vector storing the $t$ values for $\tu_{(1)}$.  These can be combined into $\tm(\tu_{(1)})$ using \cref{lem:counting_m}.  If $\tm(\tu_{(1)})$ is non-zero, then we can get $\tm(\tu_{(2)})$ simply by rescaling;  according to \cref{lem:compute_ts}, adding one to entry $t_{r(2)}$ gives the new vector of $t$'s, so only one term in the formula for $\tm$ changes in going from $U_{(1)}$ to $U_{(2)}$.  We can thus compute each $\tm(\tu_{(a)})$ in constant time, and compute all $\tm$ values in time $O(dk\log(k) + d\log(d))$.

\subsubsection{Sampling a Utility}

Given the $\tm$ values above, we sample from a slightly different distribution than the one defined in \cref{eq:utility_distribution}. The new distribution is, for $\tu_{ij} \in \tu$,
\begin{equation}
    \label{eq:alt_utility_distribution}
    \P{}{\tu_{ij}} \propto \tm(\tu_{ij})\exp\left(\frac{\eps \lceil \tu_{ij} \rceil}{2}\right).
\end{equation}
When the $\tm$ are large, sampling can be done in a numerically stable manner using logarithmic quantities (see, e.g., Appendix A.6 of~\citet{MG20}).

\subsubsection{Sampling a Sequence}

After sampling $\tu_{ij}$ from \cref{eq:alt_utility_distribution}, we sample a sequence of item indices uniformly at random from the collection of sequences with score $\tu_{ij}$. The sample fixes $s_i = j$. To sample the remaining $k-1$ items, we sample $s_1$ uniformly at random from $[t_1(\tu_{ij})] \setminus \{j\}$, $s_2$ from $[t_2(\tu_{ij})] \setminus \{s_1, j\}$, and so on. \cref{lem:t_properties} guarantees that this process never attempts to sample from an empty set.

\subsubsection{Overall Algorithm}

\joint's overall guarantees and pseudocode appear below.

\begin{figure}[t]
    \begin{algorithm}[H]
    \begin{algorithmic}[1]
       \STATE {\bfseries Input:} Vector of item counts $c_1, \ldots, c_d$, number of items to estimate $k$, privacy parameter $\eps$ \alglinelabel{algln:input}
       \STATE Sort and relabel items so $c_1 \geq c_2 \geq \cdots \geq c_d$ \alglinelabel{algln:sort_counts}
       \STATE Construct matrix $\tu$ by $\tu_{ij} = -(c_i - c_j) - \tfrac{d(k-i) + j}{2dk}$ \alglinelabel{algln:build_matrix}
       \STATE Sort $\tu$ in decreasing order to get $\tu_{(1)}, \ldots, \tu_{(dk)}$, storing the $(\text{row}, \text{column})$ of each $\tu_{(a)}$ as $(r(a), c(a))$
       \alglinelabel{algln:sort_matrix}
       \STATE Initialize $n_1, \ldots, n_k \gets 0$ \alglinelabel{algln:init_ns}
       \STATE Initialize set of non-zero $n_i$, $N \gets \emptyset$
       \STATE Initialize $b \gets 0$ \alglinelabel{algln:init_b}
       \FOR{$a = 1, \ldots, dk$} \alglinelabel{algln:fill_ns}
           \STATE $n_{r(a)} \gets c(a) - (r(a) - 1)$
           \STATE $N \gets N \cup \{r(a)\}$
           \STATE \textbf{if} $|N| = k$ \textbf{then break}
           \STATE Set $\tm(\tu_{(a)}) \gets 0$, and set $b \gets a$
       \ENDFOR
       \STATE Set $p \leftarrow \prod_{r \in [k]} n_r$ \alglinelabel{algln:first_p}
       \STATE Compute $\tm(\tu_{(b+1)}) \gets p / n_{r(a)}$
       \alglinelabel{algln:first_tm}
       \FOR{$a = b+2, \ldots, dk$} \alglinelabel{algln:loop_a}
            \STATE Set $p \gets p / n_{r(a)}$
            \STATE Compute $\tm(\tu_{(a)}) \gets p$
            \alglinelabel{algln:compute_tm}
            \STATE Update $n_{r(a)} \gets n_{r(a)} + 1$
            \STATE Update $p \gets p \cdot n_{r(a)}$
            \alglinelabel{algln:update_t}
       \ENDFOR
       \STATE Sample a utility $\tu_{ij}$ using  \cref{eq:alt_utility_distribution} \alglinelabel{algln:sample_utility}
       \STATE Initialize size-$k$ output vector $s$ with $s_i \gets j$ \alglinelabel{algln:init_output}
       \FOR{index $i' = 1, 2, \ldots, i-1, i+1, \ldots, k$} \alglinelabel{algln:loop_output}
            \STATE Compute $t_{i'}(\tu_{ij})$ by iterating through row ${i'}$ of $\tu$
            \STATE Sample $s_{i'}$ uniformly from $[t_{i'}(\tu_{ij})] \setminus \{j, s_1, s_2, \ldots, s_{i'-1}\}$ \alglinelabel{algln:uniform_sample}
       \ENDFOR
       \STATE {\bfseries Output:} Vector of item indices $s$
    \end{algorithmic}
    \caption{Efficiently sampling \joint}
    \label{alg:main}
    \end{algorithm}
\end{figure}

\begin{restatable}{theorem}{Main}
\label{thm:main}
    \joint{} samples a sequence from the exponential mechanism with utility $u^*$ in time $O(dk\log(k) + d\log(d))$ and space $O(dk)$.
\end{restatable}
\begin{proof}
    We sketch here, deferring details to \cref{sec:privacy-proof}.
    
    \textbf{Privacy}: By \cref{lem:ms}, to get $m$ it suffices to compute $\tm$\footnote{Note that even if items $i_1$ and $i_2$ have identical counts, they may have differing sequence counts, $\tm(\tu_{i_1j}) \neq \tm(\tu_{i_2j})$.  The sampling in the loop on Line~\ref{algln:loop_output} implicitly makes up for this difference.  See the full privacy proof in \cref{sec:privacy-proof}.}, and by \cref{lem:counting_m} it suffices to compute the $t$ values. A score sampled from \cref{eq:alt_utility_distribution} may be non-integral; taking its ceiling produces a utility $U_{ij} = -(c_i - c_j)$, with the desired distribution from \cref{eq:utility_distribution}.
    
    \textbf{Runtime and space}: Referring to \cref{alg:main}, line~\ref{algln:sort_counts} takes time $O(d\log(d))$ and space $O(d)$. 
    Line~\ref{algln:build_matrix} takes time and space $O(dk)$.  Line~\ref{algln:sort_matrix} takes time $O(dk\log(k))$ and space $O(dk)$; since each row of $U$ is already decreasing, we can use $k$-way merging~\cite{K97} instead of naive sorting. All remaining lines require $O(dk)$ time and space.
\end{proof}

\joint{} has the same guarantees (\cref{thm:main_privacy}, \cref{thm:main_utility}) as the exponential mechanism described at the beginning of this section, since its output distribution is identical.
\section{Experiments}
\label{sec:experiments}

Our experiments compare the peeling and joint mechanisms across several real-world datasets using the error metrics from \cref{def:errors}. All datasets and experiment code are public~\cite{Go22}. As described in \cref{subsec:related}, we only present the best pure and approximate DP baselines. Other methods are available in the experiment code. For completeness, example error plots featuring all methods appear in \cref{sec:gamma_and_laplace}.

\subsection{Comparison Methods}
\label{subsec:methods}
\subsubsection{Pure DP Peeling Mechanism}
We start with the pure DP variant, denoted \ppeel. It uses $k$ $\tfrac{\eps}{k}$-DP applications of the permute-and-flip mechanism, which dominates the exponential mechanism under basic composition (Theorem 2~\cite{MS20}). We use the equivalent exponential noise formulation~\cite{DKSSWXZ21}, where the exponential distribution $\expo{\lambda}$ is defined over $x \in \mathbb{R}$ by
\begin{equation}
\label{eq:expo}
    \P{}{x;\lambda} = \indic{x \geq 0} \cdot \lambda \cdot \exp\left(-\lambda x\right).
\end{equation}
Its pseudocode appears in \cref{alg:ppeel}. We omit the factor of 2 in the exponential distribution scale because the count utility function is monotonic (see \cref{def:monotonic} and Remark 1 of~\citet{MS20}).

\begin{algorithm}
\begin{algorithmic}[1]
   \STATE {\bfseries Input:} Vector of item counts $c_1, \ldots, c_d$, number of items to estimate $k$, privacy parameter $\eps$
   \STATE Initialize set of available items $A \gets [d]$
   \STATE Initialize empty size-$k$ output vector $s$
   \FOR{$j \in [k]$}
    \FOR{$i \in A$}
        \STATE Draw exponential noise $\eta \sim \expo{\eps/k}$
        \STATE Compute noisy count $\tilde c_i \gets c_i + \eta$
    \ENDFOR
    \STATE Set $i^* \gets \arg \max_{i \in A} \tilde c_i$
    \STATE Record chosen index, $s[j] \gets i^*$
    \STATE Remove chosen index, $A \gets A \setminus \{i^*\}$
    \ENDFOR
   \STATE {\bfseries Output:} Vector of item indices $s$
\end{algorithmic}
\caption{\ppeel, pure DP peeling mechanism}
\label{alg:ppeel}
\end{algorithm}

\begin{lemma}
\label{lem:ppeel_dp}
    \ppeel{} is $\eps$-DP.
\end{lemma}

\subsubsection{Approximate DP Peeling Mechanism}
The approximate DP variant instead uses $k$ $\eps'$-DP applications of the exponential mechanism. We do this because the exponential mechanism admits a CDP analysis that takes advantage of its bounded-range property for stronger composition; a similar analysis for permute-and-flip is not known.

We use the Gumbel-noise variant of the peeling mechanism~\cite{DR19}. This adds Gumbel noise to each raw count and outputs the sequence of item indices with the $k$ highest noisy counts. The Gumbel distribution $\gum{\beta}$ is defined over $x \in \mathbb{R}$ by
\begin{equation}
    \label{eq:gumbel}
    \P{}{x;\beta} = \frac{1}{\beta} \cdot \exp\left(-\frac{x}{\beta} - e^{-x/\beta}\right)
\end{equation}
and the resulting pseudocode appears in \cref{alg:cpeel}.

\begin{algorithm}
\begin{algorithmic}[1]
   \STATE {\bfseries Input:} Vector of item counts $c_1, \ldots, c_d$, number of items to estimate $k$, privacy parameter $\eps'$
   \FOR{$i \in [d]$}
        \STATE Draw Gumbel noise $\eta \sim \gum{k/\eps'}$
        \STATE Compute noisy count $\tilde c_i \gets c_i + \eta$
   \ENDFOR
   \STATE {\bfseries Output:} Ordered sequence of the $k$ item indices with the highest noisy counts
\end{algorithmic}
\caption{\cpeel, approx DP peeling mechanism}
\label{alg:cpeel}
\end{algorithm}

By Lemma 4.2 in~\citet{DR19}, \cpeel{} has the same output distribution as repeatedly applying the exponential mechanism and is $\eps'$-DP. A tighter analysis is possible using CDP. While an $\eps$-DP algorithm is always $\tfrac{\eps^2}{2}$-CDP, an $\eps$-DP invocation of the exponential mechanism satisfies a stronger $\tfrac{\eps^2}{8}$-CDP guarantee (Lemmas 3.2 and 3.4~\cite{CR21}). Combining this with a generic conversion from CDP to approximate DP (Proposition 1.3~\cite{BS16}) yields the following privacy guarantee:
\begin{lemma}
\label{lem:cpeel}
    \cpeel{} is $(\eps,\delta)$-DP for any $\delta > 0$ and 
    \[
        \eps = \frac{k\eps'^2}{8} + 2\eps'\sqrt{\frac{k\log(1/\delta)}{8}}.
    \]
\end{lemma}
All of our approximate DP guarantees for \cpeel{} use \cref{lem:cpeel}.

\subsection{Datasets}
\label{subsec:datasets}
We use six datasets: Books~\cite{S19} (11,000+ Goodreads books with review counts), Foods~\cite{M14} (166,000+ Amazon foods with review counts), Games~\cite{T16} (5,000+ Steam games with purchase counts), Movies~\cite{HK15} (62,000+ Movies with rating counts), News~\cite{FVC15} (40,000+ Mashable articles with share counts), and Tweets~\cite{B17} (52,000+ Tweets with like counts). For each dataset, it is reasonable to assume that one person contributes $\leq 1$ to each count, but may also contribute to many counts. Histograms of item counts appear in \cref{sec:item_count_histograms}. A more relevant quantity here is the gaps between counts of the top $k$ items (\cref{fig:diffs_and_errors}, leftmost column). As we'll see, \joint{} performs best on datasets where gaps are relatively large (Books, Movies, News, and Tweets).

\subsection{Results}
\label{subsec:results}
The experiments evaluate error across the three mechanisms, six datasets, and three error metrics. For each mechanism, the center line plots the median error from 50 trials (padded by 1 to avoid discontinuities on the logarithmic $y$-axis), and the shaded region spans the $25^{th}$ to $75^{th}$ percentiles. We use $k = 5, 15, \ldots, 195$ with 1-DP instances of \joint{} and \ppeel{} and $(1,10^{-6})$-DP instances of \cpeel. Due to the weakness of the $k$-relative error metric, and for the sake of space in the figure, we relegate its discussion to \cref{sec:k_relative_errors}.

\arxiv{\begin{figure*}[p]
    \centering
    \includegraphics[scale=0.34,trim={0 -1.1cm 0 0},clip]{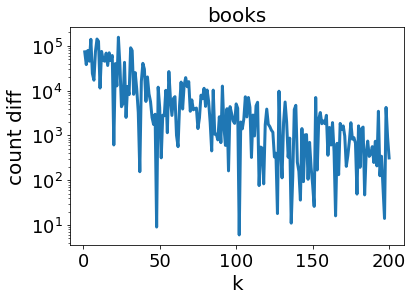}
    \includegraphics[scale=0.34]{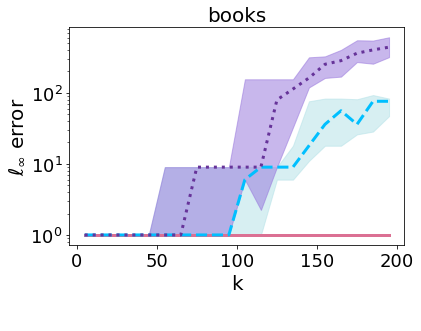}
    \includegraphics[scale=0.34]{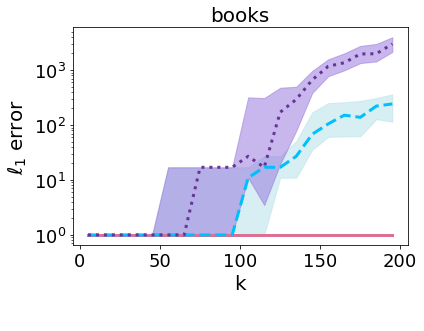} \\ \vspace{-10pt}
    \includegraphics[scale=0.34,trim={0 -1.1cm 0 0},clip]{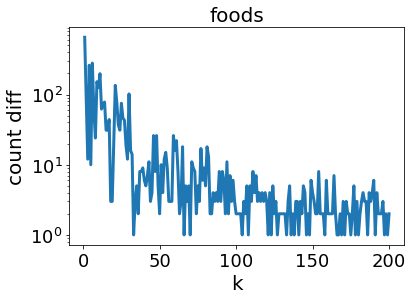}
    \includegraphics[scale=0.34]{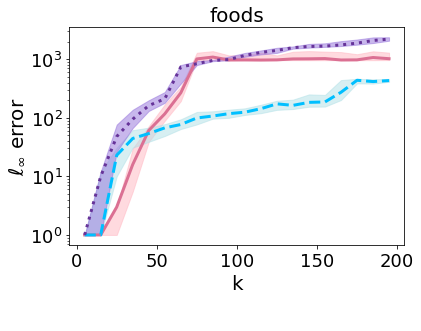}
    \includegraphics[scale=0.34]{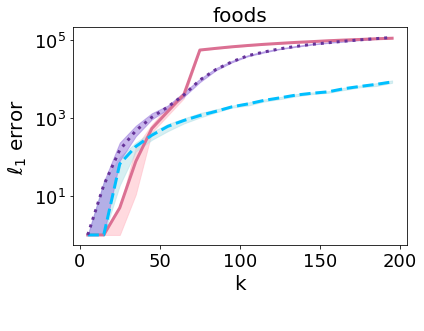} \\ \vspace{-10pt}
    \includegraphics[scale=0.34,trim={0 -1.1cm 0 0},clip]{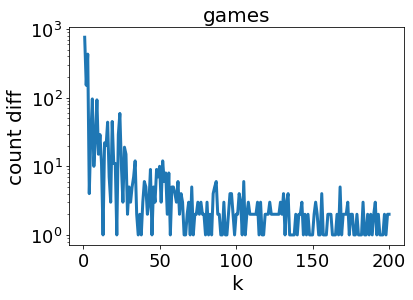}
    \includegraphics[scale=0.34]{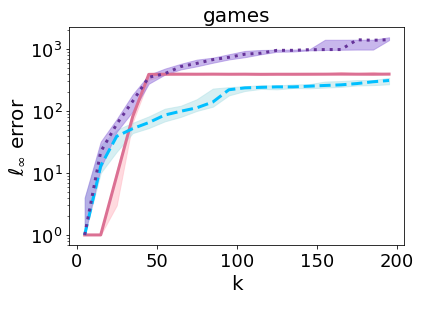}
    \includegraphics[scale=0.34]{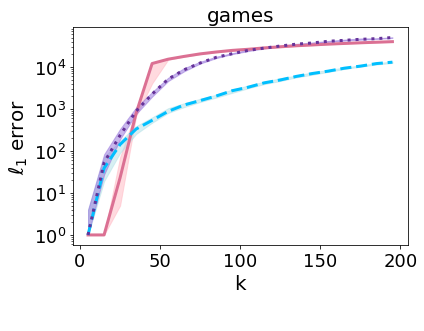} \\ \vspace{-10pt}
    \includegraphics[scale=0.34,trim={0 -1.1cm 0 0},clip]{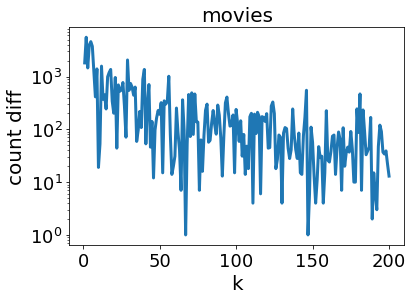}
    \includegraphics[scale=0.34]{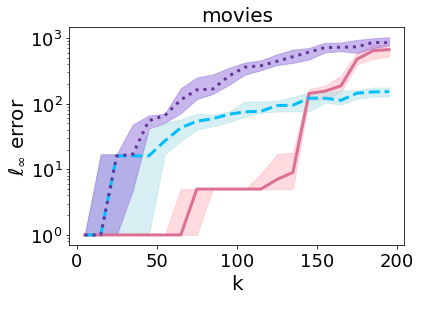}
    \includegraphics[scale=0.34]{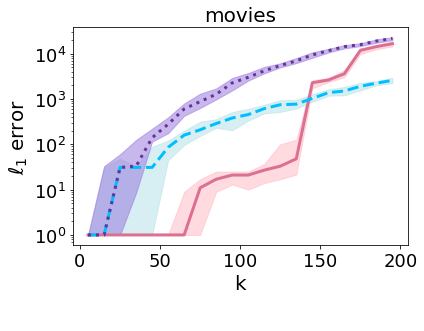} \\ \vspace{-10pt}
    \includegraphics[scale=0.34,trim={0 -1.1cm 0 0},clip]{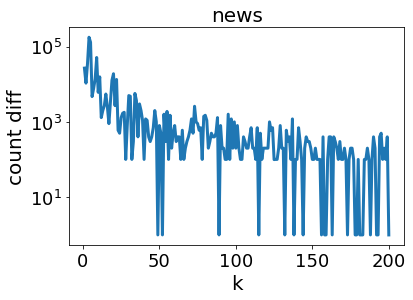}
    \includegraphics[scale=0.34]{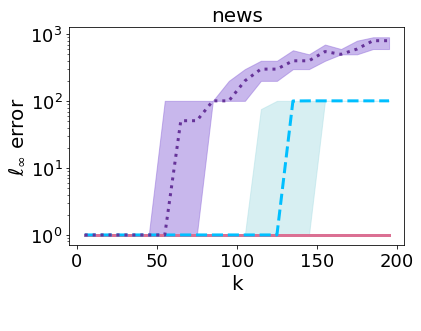}
    \includegraphics[scale=0.34]{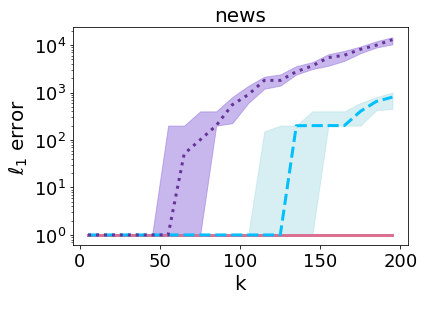} \\ \vspace{-10pt}
    \includegraphics[scale=0.34,trim={0 -1.1cm 0 0},clip]{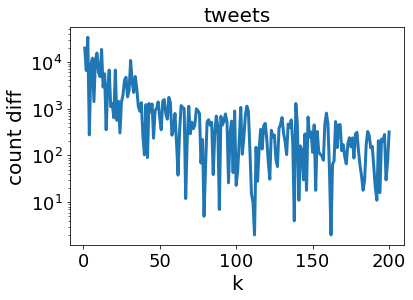}
    \includegraphics[scale=0.34]{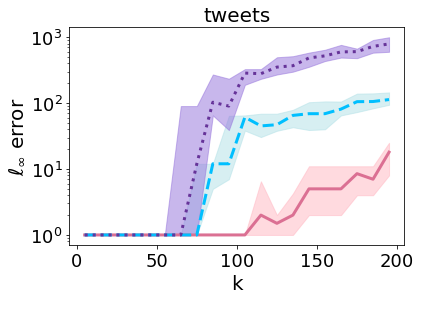}
    \includegraphics[scale=0.34]{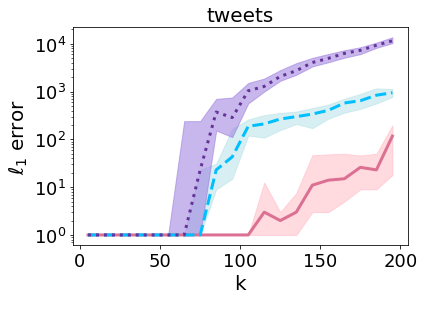} \\ \vspace{-10pt}
    \includegraphics[scale=0.43]{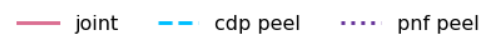}
    \caption{Note that all plots have a logarithmic $y$-axis and quantities are padded by 1 to avoid discontinuities on the logarithmic $y$-axis.  \textbf{Left column}: Count differences $c_k - c_{k+1}$ for each dataset.  \textbf{Center column}: $\ell_\infty$ error.  \textbf{Right column}: $\ell_1$ error.}
    \label{fig:diffs_and_errors}
\end{figure*}}

\subsubsection{$\ell_\infty$ error}
\joint's performance is strongest for $\ell_\infty$ error (\cref{fig:diffs_and_errors}, center column). This effect is particularly pronounced on the Books, Movies, News, and Tweets datasets. This is because these datasets have large gaps between the top $k$ counts (\cref{fig:diffs_and_errors}, leftmost column), which results in large gaps between the scores that \joint{} assigns to optimal and suboptimal sequences. These large gaps enable \joint{} to obtain much stronger performance than the baseline pure DP algorithm, \ppeel, and to beat even the approximate DP \cpeel{} for a wide range of $k$. In contrast, small gaps reduce this effect on Foods and Games. On these datasets, \joint{} slightly improves on \ppeel{} overall, and only improves on \cpeel{} for roughly $k \leq 30$.

\narxiv{\begin{figure}[h!]
    \centering
    \begin{subfigure}{0.23\textwidth}
        \includegraphics[scale=0.27]{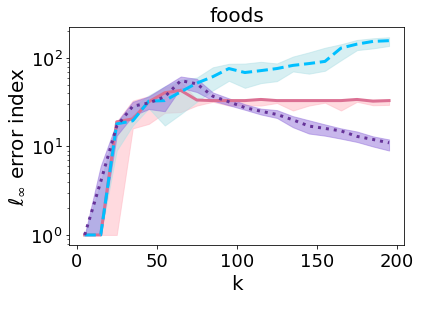} 
    \end{subfigure}
    \vspace{-10pt}
    \begin{subfigure}{0.23\textwidth}
        \includegraphics[scale=0.27]{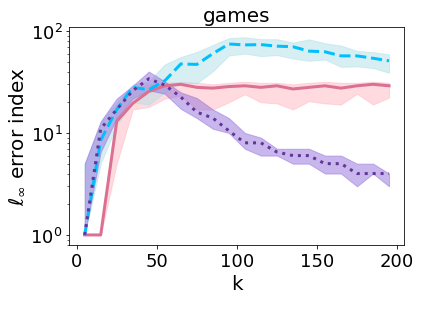}
    \end{subfigure}
    \vspace{-5pt}
    \begin{subfigure}{0.46\textwidth}
        \centering \includegraphics[scale=0.33]{images/some_labels.png}
    \end{subfigure}
    \caption{$\ell_\infty$ index error plots. Note the logarithmic $y$-axis.}
    \label{fig:linf_idx_error}
\end{figure}}

\arxiv{\begin{figure}[h!]
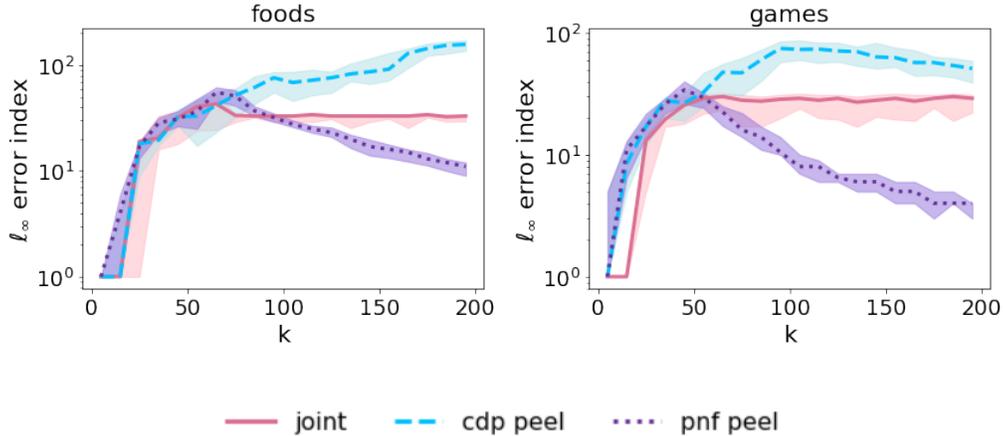

    \centering
    \begin{subfigure}{0.4\textwidth}
        \includegraphics[scale=0.45]{images/foods_linf_idx.png} 
    \end{subfigure}
    \begin{subfigure}{0.4\textwidth}
        \includegraphics[scale=0.45]{images/games_linf_idx.png}
    \end{subfigure}
    \begin{subfigure}{0.8\textwidth}
        \centering \includegraphics[scale=0.45]{images/some_labels.png}
    \end{subfigure}
    \caption{$\ell_\infty$ index error plots. Note the logarithmic $y$-axis.}
    \label{fig:linf_idx_error}
\end{figure}}

The $\ell_{\infty}$ metric also features plateaus in \joint's error on Foods and Games. This is because \joint's error is gap-dependent while \ppeel{} and \cpeel's errors are more $k$-dependent: as $k$ grows, \joint's maximum error may change as it ranks more items, but the item index where that error occurs changes monotonically. The reason is that \joint's error ultimately depends on the count gaps under consideration. In contrast, the item index where \ppeel{} and \cpeel{} incur maximum error may increase and then decrease. This is because \ppeel{} and (to a lesser extent) \cpeel{} must divide their privacy budget by $k$, and thus are increasingly likely to err (and incur the larger penalties for) top items as $k$ becomes large. \cref{fig:linf_idx_error} plots the maximum error item index and illustrates this effect.

\subsubsection{$\ell_1$ error}
A similar trend holds for $\ell_1$ error (\cref{fig:diffs_and_errors}, rightmost column). \joint{} again largely obtains the best performance for the Books, Movies, News, and Tweets datasets, with relatively worse error on Foods and Games. $\ell_1$ error is a slightly more awkward fit for \joint{} because \joint's utility function relies on maximum count differences; \joint{} thus applies the same score to sequences where a single item count has error $c$ and sequences where every item count has error $c$. This means that \joint{} selects sequences that have relatively low maximum (and $\ell_\infty$) error but may have high $\ell_1$ error. Nonetheless, we again see that \joint{} always obtains the strongest performance for small $k$; it matches \ppeel{} for small datasets and outperforms it for large ones; and it often outperforms \cpeel{}, particularly for large datasets and moderate $k$.

\subsubsection{Time comparison}
We conclude with a time comparison using the largest dataset (Foods, $d \approx 166,000$) and 5 trials for each $k$. \ppeel{} uses $k$ instances of the permute-and-flip mechanism for an overall runtime of $O(dk)$. \cpeel{}'s runtime is dominated by finding the top-$k$ values from a set of $d$ unordered values, which can be done in time $O(d + k\log(k))$. As seen in \cref{fig:time}, and as expected from their asymptotic runtimes, \joint{} is slower than \ppeel, and \ppeel{} is slower than \cpeel. Nonetheless, \joint{} still primarily runs in seconds or, for $k = 200$, slightly over 1 minute.

\narxiv{\begin{figure}[h]
    \begin{subfigure}{0.46\textwidth}
        \centering \includegraphics[scale=0.33]{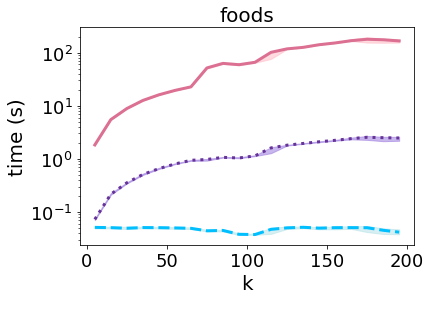}
    \end{subfigure}
    \begin{subfigure}{0.46\textwidth}
    \vspace{-12pt}
        \centering \includegraphics[scale=0.3]{images/some_labels.png}
    \end{subfigure}
    \vspace{-10pt}
    \caption{Execution time on a logarithmic $y$-axis.}
    \label{fig:time}
\end{figure}}
\arxiv{\begin{figure}[h]
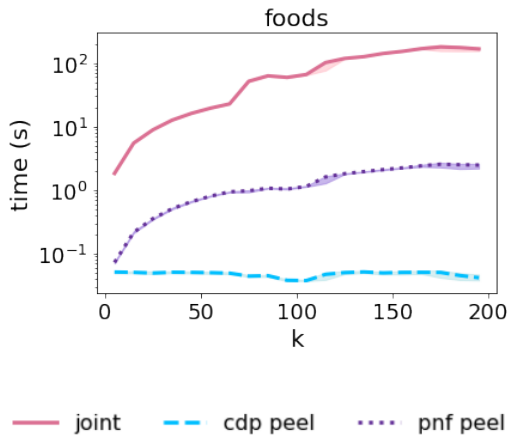

    \begin{subfigure}{\textwidth}
        \centering \includegraphics[scale=0.45]{images/foods_time}
    \end{subfigure} \\
    \begin{subfigure}{\textwidth}
        \centering \includegraphics[scale=0.4]{images/some_labels.png}
    \end{subfigure}
    \caption{Execution time on a logarithmic $y$-axis.}
    \label{fig:time}
\end{figure}}
\section{Conclusion}
\label{sec:conclusion}

We defined a joint exponential mechanism for the problem of differentially private top-$k$ selection and derived an algorithm for efficiently sampling from its distribution.  We provided code and experiments demonstrating that our approach almost always improves on existing pure DP methods and often improves on existing approximate DP methods when $k$ is not large. We focused on the standard setting where an individual user can contribute to all item counts.  However, if users are restricted to contributing to a single item, then algorithms that modify item counts via Laplace noise~\cite{DWZK19, QSZ21} are superior to \joint{} and peeling mechanisms.  The best approach for the case where users can contribute to some number of items larger than $1$ but less than $d$ is potentially a topic for future work.

\section*{Acknowledgements}
We thank Ryan Rogers for helpful discussion of the peeling mechanism.

\newpage

\bibliography{arxiv_main}
\bibliographystyle{plainnat}

\appendix
\onecolumn

\section{Proof of \joint{} Utility Guarantee (\cref{thm:main_utility})}
\label{sec:utility-proof}

\begin{proof}[Proof of \cref{thm:main_utility}]
The following is a basic utility guarantee for the exponential mechanism.
\begin{lemma}[\citet{MT07, DR14}]
    Let $\calA(u,D)$ be the utility value produced by an instance of the exponential mechanism with score function $u$, output space $R$, dataset $D$, and optimal utility value $OPT_u(D)$. Then 
    \[
        \P{}{\calA(u,D) \leq OPT_u(D) - \frac{2\Delta(u)}{\eps}[\ln(|R|) + t]} \leq e^{-t}. 
    \]
\end{lemma}
Taking $t = 5$, and using the fact that $|R| \leq d^k$ for \joint 's utility function $u^*$, completes the result.
\end{proof}

\section{Full Privacy, Runtime, and Storage Space Proof For \joint{} (\cref{thm:main})}
\label{sec:privacy-proof}

\begin{proof}[Proof of \cref{thm:main}]
Recall that \joint{} refers to the algorithm that uses the efficient sampling mechanism. Here, we first prove that \joint{} samples a sequence from the exponential mechanism with utility $u^*$.

Let \emech{} refer to the naive original construction of the exponential mechanism with utility $u^*$. It suffices to show that \joint{} and \emech{} have identical output distributions. Fix some sequence $S = (s_1, \ldots, s_k)$ of indices from $[d]$.

If $s_1, \ldots, s_k$ are not distinct, then \joint{} never outputs $S$. This agrees with the original definition of the exponential mechanism with utility function $u^*$, which assigns score $-\infty$ to any sequence of item indices with repetitions. Thus, for any $S$ with non-distinct elements, $\P{\joint}{\text{output } S} = \P{\emech}{\text{output } S} = 0$.

If instead $s_1, \ldots, s_k$ are distinct, by \cref{lem:counting_m}, $\tm(S) > 0$. Let $\tu_{i^*j^*} = \min_{i \in [k]} -(c_i - c_{s_i}) - z_{is_i}$ be its score in $\tu$, so $i^* = \arg \min_{i \in [k]} -(c_i - c_{s_i}) - z_{is_i}$. Let $U_Z = \{-(c_i - c_j)\}_{i \in [k], j \in [d]}$ denote the set of possible values for $-(c_i - c_j)$; note that this a set of integers and does not have repeated elements. Then
\begin{align*}
    \P{\joint}{\text{output } S} =&\ \P{\joint}{\text{sample score } \tu_{i^*j^*}} \cdot \P{\joint}{\text{sample sequence } S \mid \text{sample score } \tu_{i^*j^*}} \\
    =&\ \frac{\tm(\tu_{i^*j^*}) \exp \left(\frac{\eps \lceil \tu_{i^*j^*} \rceil}{2}\right)}{\sum_{u \in \tu} \tm(u) \exp \left(\frac{\eps \lceil u \rceil}{2}\right)} \cdot \prod_{r \neq i^*} \frac{1}{t_r(\tu_{i^*j^*}) - (r - 1)} \\
    =&\ \frac{\exp \left(\frac{\eps \lceil \tu_{i^*j^*} \rceil}{2}\right)}{\sum_{u \in \tu} \tm(u) \exp \left(\frac{\eps \lceil u \rceil}{2}\right)}
\end{align*}
by \cref{lem:counting_m}. Then we continue the chain of equalities as 
\begin{align*}
    \frac{\exp \left(\frac{\eps \lceil \tu_{i^*j^*} \rceil}{2}\right)}{\sum_{u \in \tu} \tm(u) \exp \left(\frac{\eps \lceil u \rceil}{2}\right)} =&\ \frac{\exp \left(\frac{\eps \lceil \tu_{i^*j^*} \rceil}{2}\right)}{\sum_{A_{ij}} \sum_{u \in A_{ij}} \tm(u) \exp \left(\frac{\eps \lceil u \rceil}{2}\right)} \\
    =&\ \frac{\exp \left(\frac{\eps \lceil \tu_{i^*j^*} \rceil}{2}\right)}{\sum_{u \in U_Z} m(u) \exp \left(\frac{\eps u}{2}\right)} \\
    =&\ \frac{m(\lceil \tu_{i^*j^*} \rceil) \exp \left(\frac{\eps \lceil \tu_{i^*j^*} \rceil}{2}\right)}{\sum_{u \in U_Z} m(u) \exp \left(\frac{\eps u}{2}\right)} \cdot \frac{1}{m(\lceil \tu_{i^*j^*} \rceil)} \\
    =&\ \P{\emech}{\text{sample score } \lceil \tu_{i^*j^*} \rceil} \cdot \P{\emech}{\text{sample sequence } S \mid \text{sample score } \lceil \tu_{i^*j^*} \rceil} \\
    =&\ \P{\emech}{\text{output } S}
\end{align*}
where the second equality uses \cref{lem:ms}.

Having established the privacy of \joint{}, we now turn to proving that its runtime and storage space costs are $O(dk\log(k) + d\log(d))$ and $O(dk)$, respectively.

Referring to \cref{alg:main}, line~\ref{algln:sort_counts} takes time $O(d\log(d))$ and space $O(d)$. Line~\ref{algln:build_matrix} takes time and space $O(dk)$. Line~\ref{algln:sort_matrix} takes time $O(dk\log(k))$ and space $O(dk)$; since each row of $U$ is already decreasing, we can use $k$-way merging~\cite{K97} instead of naive sorting.

The loop on Line~\ref{algln:fill_ns} handles the $\tm$ that are zero.  Its variable setup on Lines~\ref{algln:init_ns}-\ref{algln:init_b} takes time and space $O(k)$.  Lines internal to the loop each take $O(1)$ time and space.  So, overall, this block of code requires time and space $O(dk)$.
    
The loop on Line~\ref{algln:loop_a} handles the non-zero $\tm$.  Its variable setup on Lines~\ref{algln:first_p}-\ref{algln:first_tm} takes time $O(k)$ and space $O(1)$.  Lines internal to the loop each take $O(1)$ time and space. So, overall, this block of code requires time and space $O(dk)$.
    
Sampling a utility (Line~\ref{algln:sample_utility}) requires time and space $O(dk)$.  The remaining loop (Line~\ref{algln:loop_output}) iterates for $O(k)$ steps, and each step requires $O(d)$ time and space.

Overall, this yields runtime and storage space costs of $O(dk\log(k) + d\log(d))$ and $O(dk)$, respectively.
\end{proof}

\section{Other Experiment Plots}
\label{sec:experiment_other}

\subsection{Item Count Histograms}
\label{sec:item_count_histograms}

\cref{fig:mhistograms} contains item count histograms for each of the datasets.

\begin{figure}[h!]
    \centering
    \begin{subfigure}{0.3\textwidth}
        \includegraphics[scale=0.35]{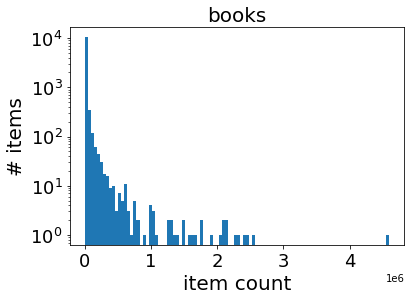}
    \end{subfigure}
    \begin{subfigure}{0.3\textwidth}
        \includegraphics[scale=0.35]{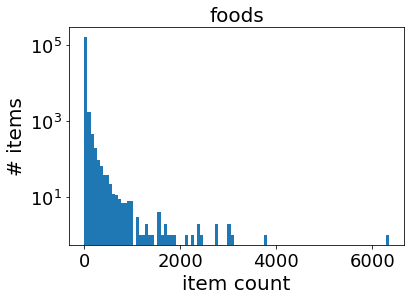}
    \end{subfigure}
    \begin{subfigure}{0.3\textwidth}
        \includegraphics[scale=0.35]{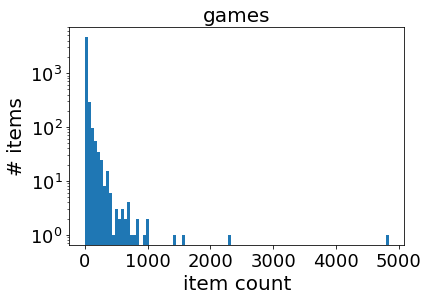}
    \end{subfigure}
        \begin{subfigure}{0.3\textwidth}
        \includegraphics[scale=0.35]{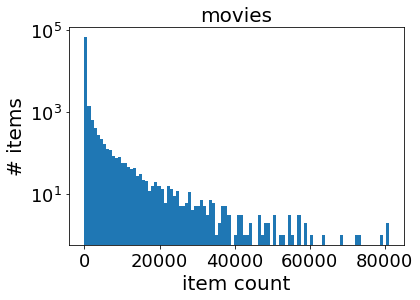}
    \end{subfigure}
        \begin{subfigure}{0.3\textwidth}
        \includegraphics[scale=0.35]{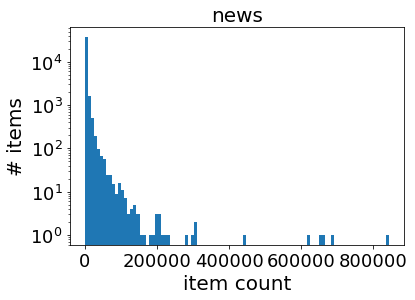}
    \end{subfigure}
        \begin{subfigure}{0.3\textwidth}
        \includegraphics[scale=0.35]{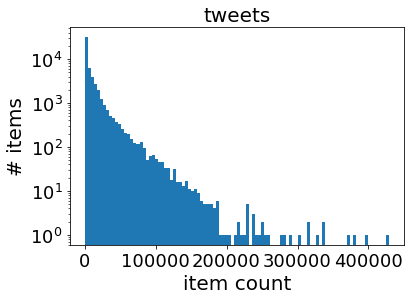}
    \end{subfigure}
    \caption{Item count histograms. The $x$-axis is binned by item count, and the $y$-axis is the number of items in each bin.}
    \label{fig:mhistograms}
\end{figure}

\subsection{$k$-Relative Errors}
\label{sec:k_relative_errors}

\cref{fig:krel_error} plots $k$-relative error for each of the mechanisms and datasets. The trends for $k$-relative error are broadly unchanged from $\ell_\infty$ and $\ell_1$ error: \joint{} consistently matches or outperforms its pure DP counterpart \ppeel, mostly outperforms \cpeel{} on large-scale datasets, and is mostly outperformed by \cpeel{} on small-scale datasets unless $k$ is small. However, $k$-relative error is the least sensitive error (see discussion after \cref{def:errors}), so for several datasets the performance gaps between methods are small or zero.

\begin{figure}[h!]
    \centering
    \begin{subfigure}{0.3\textwidth}
        \includegraphics[scale=0.33]{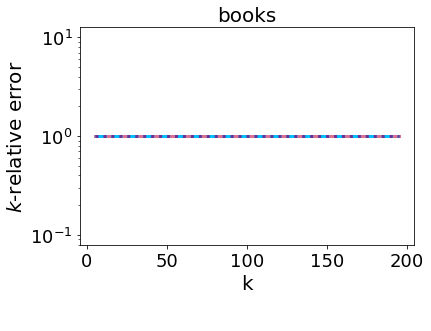}
    \end{subfigure}
    \begin{subfigure}{0.3\textwidth}
        \includegraphics[scale=0.33]{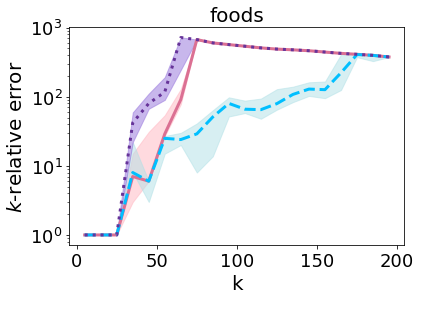} 
    \end{subfigure}
    \begin{subfigure}{0.3\textwidth}
        \includegraphics[scale=0.33]{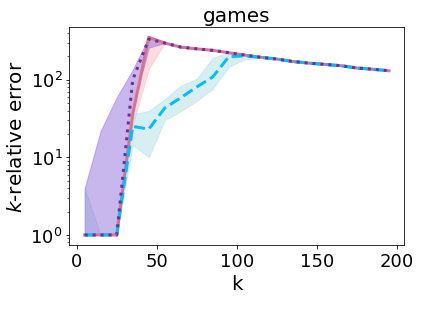}
    \end{subfigure}
    \begin{subfigure}{0.3\textwidth}
        \includegraphics[scale=0.33]{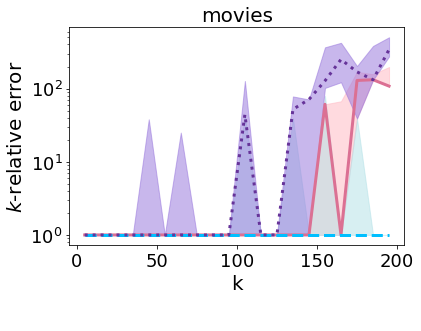}
    \end{subfigure}
    \begin{subfigure}{0.3\textwidth}
        \includegraphics[scale=0.33]{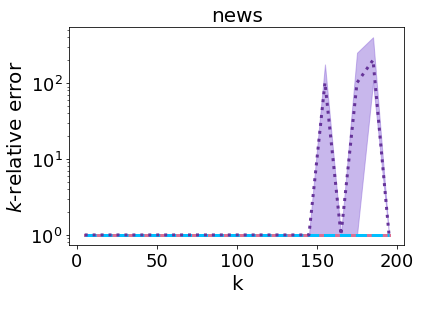}
    \end{subfigure}
    \vspace{-10pt}
    \begin{subfigure}{0.3\textwidth}
        \includegraphics[scale=0.35]{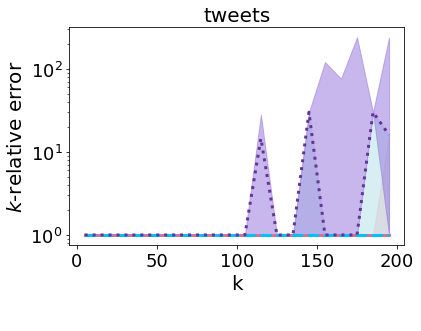}
    \end{subfigure}
    \begin{subfigure}{0.46\textwidth}
        \centering \includegraphics[scale=0.4]{images/some_labels.png}
    \end{subfigure}
    \caption{$k$-relative error plots. Note the logarithmic $y$-axis.}
    \label{fig:krel_error}
\end{figure}

\subsection{Gamma and Laplace Mechanisms}
\label{sec:gamma_and_laplace}

\cref{fig:gamma_laplace} plots error on the Movies dataset for the core three methods (\joint{}, \ppeel{}, and \cpeel{}) as well as Gamma (Theorem 4.1, \cite{SU15}) and Laplace~\cite{DWZK19} mechanisms, which are respectively dominated by \ppeel{} and \cpeel{}.  Exact details of these mechanisms can be found in the code provided in the supplement.  Note that the Laplace mechanism of \citet{QSZ21} is identical to that of \citet{DWZK19} at the tested value of $\epsilon = 1$.  For $\epsilon < 0.1$, \citet{QSZ21} also provides an approximate-DP version of the algorithm, which may be more competitive in that setting.

\begin{figure}[h!]
    \centering
    \includegraphics[scale=0.35]{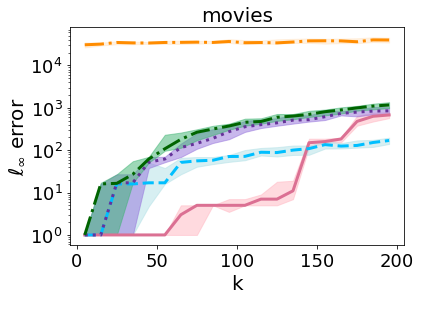}
    \includegraphics[scale=0.35]{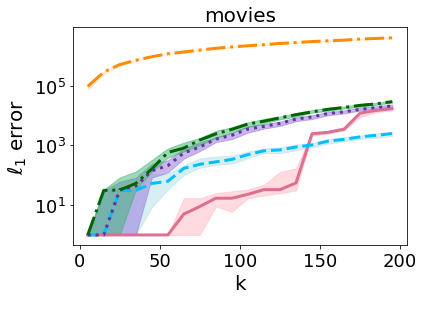}
    \includegraphics[scale=0.35]{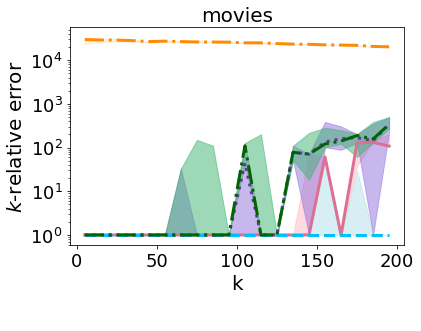}
    \centering \includegraphics[scale=0.35]{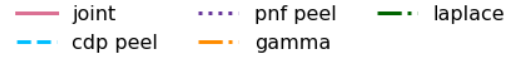}
    \caption{Plots for all error metrics on the Movies dataset, with Gamma and Laplace mechanisms included as additional baselines. Note the logarithmic $y$-axis.}
    \label{fig:gamma_laplace}
\end{figure}

\arxiv{
\section{Permute-and-Flip Version of \joint}
\label{sec:pnf}

The permute-and-flip mechanism of~\citet{MS20}, which is equivalent to report-noisy-max with exponential noise~\citep{DKSSWXZ21}, strictly dominates the exponential mechanism in terms of expected utility.  In theory, the utility gain can be as large as a factor of 2.  In this section, we show how a minor modification to \joint{} can transform it from an exponential mechanism into a permute-and-flip mechanism.  However, we then show that the empirical performance of this alternative mechanism is not significantly different than that of the exponential mechanism for any of the datasets considered in this work.

\subsection{Sketch of the \pnfjoint{} algorithm}

We derive the modification to the joint exponential mechanism based on the report-noisy-max with exponential noise formulation of permute-and-flip.  We will call the resulting algorithm \pnfjoint. 
 \cref{alg:rnmexpo} recalls the generic version of report-noisy-max with exponential noise as given in~\citet{DKSSWXZ21}.

\begin{figure}[H]
    \begin{algorithm}[H]
    \begin{algorithmic}[1]
       \STATE {\bfseries Input:} Dataset $D$, utility $u$ with $\ell_1$ sensitivity $\Delta(u)$, set of outputs $O$, privacy parameter $\eps$ \alglinelabel{algln:rnmexpo_in}
       \FOR{$o \in O$} \alglinelabel{algln:rmnexpo_loop_outputs}
           \STATE $v_o \gets u(D, o) + \mathrm{Expo}\left(\frac{\epsilon}{2 \Delta(u)}\right)$
       \ENDFOR
       \STATE {\bfseries Output:} $\argmax_{o \in O} v_o$ \alglinelabel{algln:rnmexpo_out}
    \end{algorithmic}
    \caption{\rnmexpo, \citet{DKSSWXZ21}}
    \label{alg:rnmexpo}
    \end{algorithm}
\end{figure}

We cannot typically afford to naively run \rnmexpo{} as written, looping over all outputs, since the set of outputs $O$ is exponentially large for the joint top-$k$ problem ($O(d^k)$).  Yet, if we use the utility function $u^*$ defined in \cref{sec:joint}, recall that there are then only $O(dk)$ unique utility values.  As with the exponential mechanism, we can leverage this fact to yield an efficient sampling algorithm.

More specifically, we can simply switch out the utility-sampling step of \joint{}, which applies \cref{eq:alt_utility_distribution} (see Line~\ref{algln:sample_utility} of \cref{alg:main}).  To see this, recall that \joint{} first computes $dk$ counts: the count $m(U_{(a)})$ represents the number of outputs whose utility is $U_{(a)}$.  Then, a specific utility value is selected according to \cref{eq:alt_utility_distribution}.  Finally, a random sequence with this utility is returned.  To convert from an exponential mechanism to report-noisy-max with exponential noise, we simply need to replace the utility value selection step.  That is, we need to select a utility value $U_{(a)}$ proportional to the probability that one of its sequences will have the maximizing noisy value on Line~\ref{algln:rnmexpo_out} of \cref{alg:rnmexpo}.  The utility value $U_{(a)}$ has $m(U_{(a)})$ chances to be the maximizer.  If we compute the maximum over $m(U_{(a)})$ draws from the exponential distribution $\mathrm{Expo}(\epsilon / (2 \Delta(u)))$, then this tells us the max $v_o$ value of any output $o$ with utility $U_{(a)}$.  Hence, to decide whether a sequence with utility $U_{(a)}$ will be output by \cref{alg:rnmexpo}, we just need to know what this maximum value is.  This implies that, rather than actually having to draw from the exponential distribution $m(U_{(a)})$ times (which would be impractical given that $m$ can be $O(d^k)$), we can draw once from the distribution of the maximum.  We will call this distribution \maxexpo.

\cref{alg:pnfjoint} summarizes this process.  Note that, assuming \maxexpo{} can be sampled in $O(1)$ time and space, this algorithm runs in $O(dk)$ time and $O(1)$ space.  Hence, the overall \pnfjoint{} algorithm to draw a top-$k$ sequence has the same time and space complexity as the \joint{} exponential mechanism.

\begin{figure}[H]
    \begin{algorithm}[H]
    \begin{algorithmic}[1]
       \STATE {\bfseries Input:} Unique utility values $\{U_{(1)}, \ldots U_{(dk)}\}$, counts of the number of sequences with each utility value $\{m(U_{1}), \ldots, m(U_{(dk)}\}$, privacy parameter $\eps$ \alglinelabel{algln:input}
       \FOR{$a \in [dk]$} \alglinelabel{algln:loop_outputs}
           \STATE $v_a \gets U_{(a)} + \maxexpo\left(\epsilon/2, m(U_{(a)})\right)$
       \ENDFOR
       \STATE {\bfseries Output:} $\argmax_{a \in [dk]} v_a$ \alglinelabel{algln:rnmexpo_out}
    \end{algorithmic}
    \caption{\pnfjoint{} utility selection component}
    \label{alg:pnfjoint}
    \end{algorithm}
\end{figure}

The $\maxexpo$ distribution has a simple CDF that can be derived from that of the exponential distribution.  Specifically, suppose that we want to know the maximum over $m$ draws from $\mathrm{Expo}(\epsilon / (2 \Delta(u)))$.  Letting $X_i$ represent a draw from this $\mathrm{Expo}$, we have:
\begin{equation}
    F(z) \coloneqq P\left(\max_{i \in [m]} X_i \leq z \right) = \prod_{i \in [m]} P(X_i \leq z) = \left(1 - \exp\left(-\frac{\epsilon z}{2 \Delta(u)}\right)\right)^m
\end{equation}
where the first equality follows from the fact that the draws are independent, and the second from the definition of the CDF of the exponential distribution.  Inverting this CDF, we have:
\begin{equation}
    \label{eq:inv_cdf}
    F^{-1}(p) = -\frac{2 \Delta(u)}{\epsilon} \ln\left(1 - p^{1/m}\right).
\end{equation}
Drawing a sample from the uniform distribution on $[0, 1]$ and plugging it in for $p$ above produces a sample $z$ from the desired distribution.

\subsection{Experiments with \pnfjoint}

The code implementing \pnfjoint{} is publicly available~\citep{Go22}.  We note though that this implementation is not as numerically stable as that of \joint{}, so it cannot handle the largest values of $k$ used in our other experiments.  The main difference is that, while \joint{} can be implemented using only the logs of the $m$ values, \pnfjoint{} needs the actual $m$ (or $1/m$) values for the \maxexpo{} inverse CDF (\cref{eq:inv_cdf}).  If $m$ is too large to be stored in a double, or if $1/m$ is so small that its double representation is zero, then this poses an implementation challenge\footnote{In fact, even using log counts (stored as doubles) in \joint{} is technically numerically imperfect, as low-order bits of the counts may be lost.  However, these should be negligible compared to the magnitude of the overall count.  To patch this imperfection, \joint{} could instead be implemented to store exact counts in an arbitrary-precision integer data type.}.

\cref{fig:pnfjoint_experiments} shows empirical results.  For each mechanism, the center line plots the median error from 50 trials (padded by 1 to avoid discontinuities on the logarithmic $y$-axis), and the shaded region spans the $25^{th}$ to $75^{th}$ percentiles.  Only the low-$k$ regime is assessed, as this is where the implementation of \pnfjoint{} is numerically stable.  In this regime, for three of the datsets all methods have essentially zero error.  For the other three datasets, the results for \joint{} and \pnfjoint{} are almost indistinguishable.

\begin{figure*}[ht]
    \centering
    \includegraphics[scale=0.34,trim={2cm 2cm 2cm 0},clip]{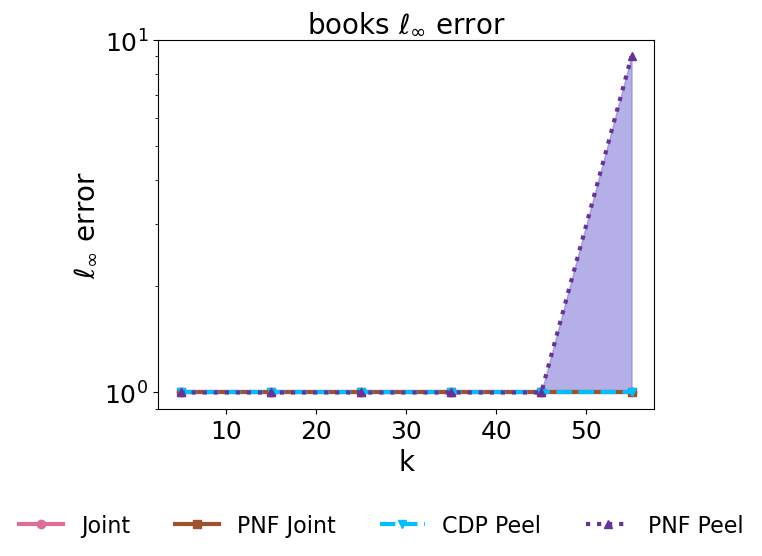}
    \includegraphics[scale=0.34,trim={2cm 2cm 2cm 0},clip]{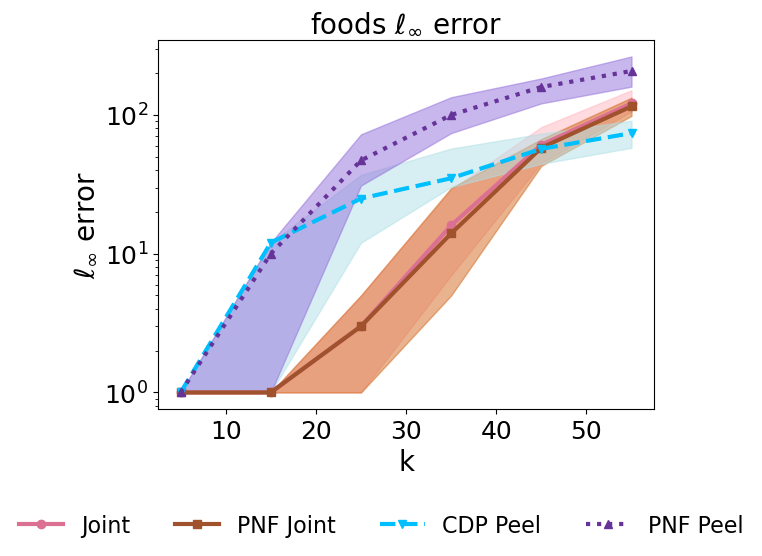}
    \includegraphics[scale=0.34,trim={2cm 2cm 2cm 0},clip]{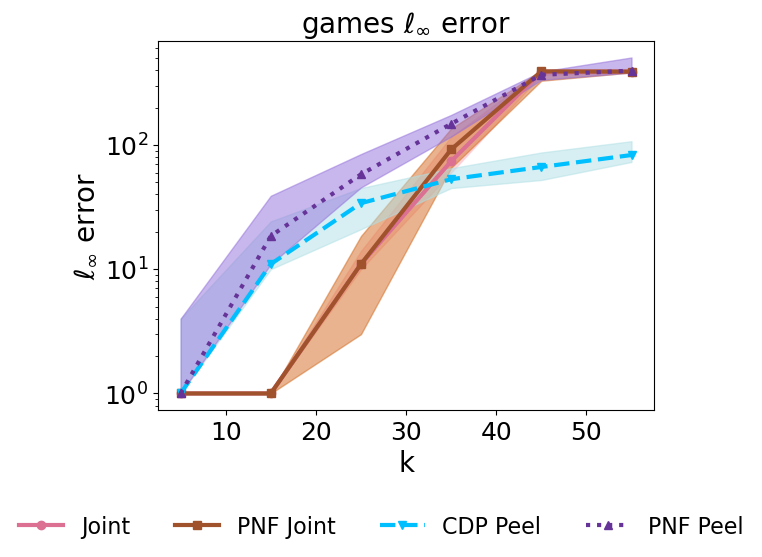} \\ \vspace{+10pt}
    \includegraphics[scale=0.34,trim={2cm 2cm 2cm 0},clip]{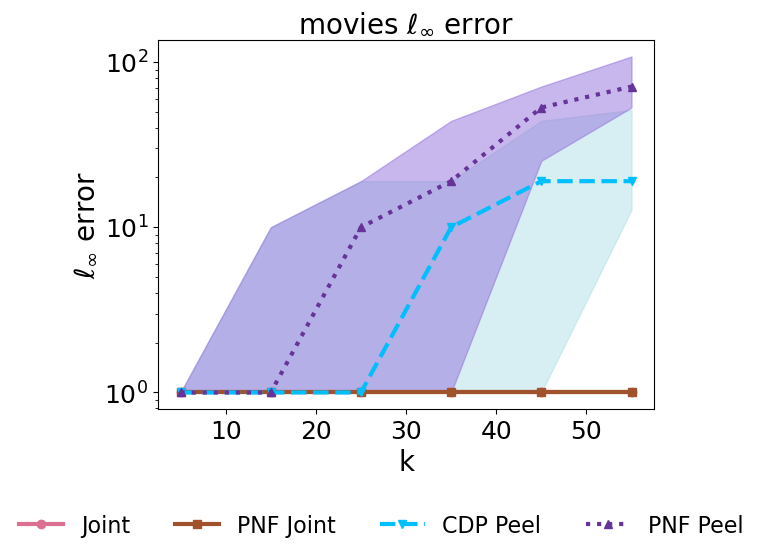}
    \includegraphics[scale=0.34,trim={2cm 2cm 2cm 0},clip]{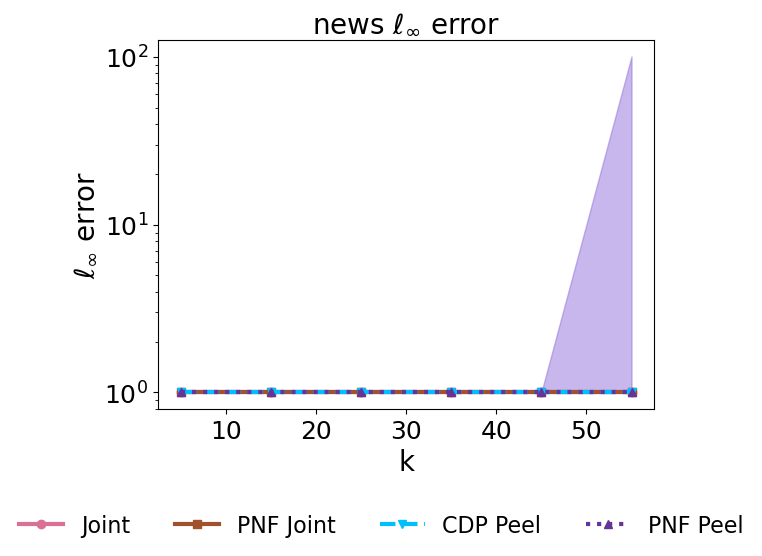}
    \includegraphics[scale=0.34,trim={2cm 2cm 2cm 0},clip]{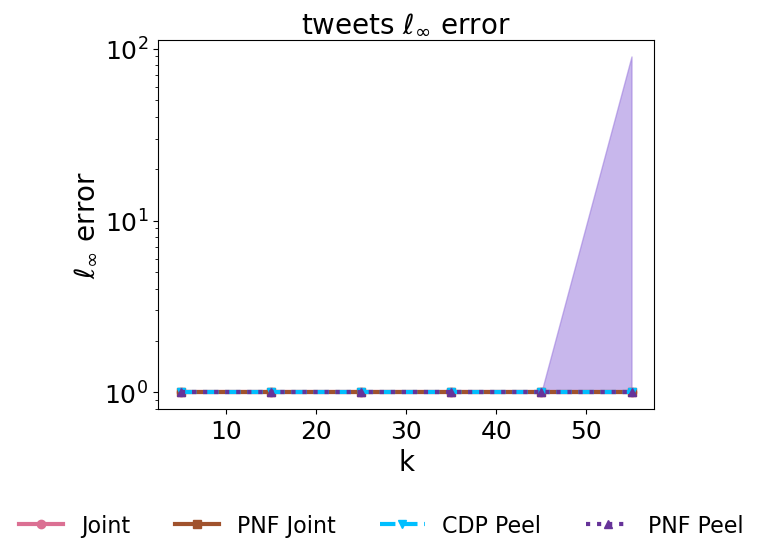} \\
    \includegraphics[scale=0.43]{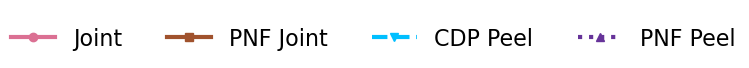}
    \caption{$\ell_\infty$ error for $k \leq 45$.  Note the logarithmic $y$-axis.}
    \label{fig:pnfjoint_experiments}
\end{figure*}

\subsection{Discussion}

The indistinguishably of \joint{} and \pnfjoint{} for these real-world datasets is perhaps to be expected.  Despite the fact that \pnfjoint{} can in theory have utility twice that of \joint{} for some applications, there tends not to be much difference when the output space is large and there are many high-utility outputs.

To understand this better, it helps to consider Algorithm 2 of the original permute-and-flip work~\citep{MS20}.  Their Algorithm 2 is reproduced here as \cref{alg:exp_mech_as_rejection_sampling}.  This algorithm is an implementation of the exponential mechanism, written as rejection sampling.  It is completely identical to the permute-and-flip mechanism, except that permute-and-flip adds one additional line inside the loop that removes output $o$ from consideration once it has been rejected once: $O \gets O \setminus {o}$.  Hence, the exponential mechanism can be thought of as identical to the permute-and-flip mechanism, except that the former samples from the space of outputs \emph{with} replacement, while the latter samples \emph{without} replacement.  If the space of outputs is very large, and a nontrivial number of outputs have relatively high probability, then sampling with replacement is not going to be much different from sampling without replacement.  That is, the chance that an output will be re-sampled by the exponential mechanism is low.  In our joint top-$k$ formulation, the output space is indeed very large: $O(d^k)$, with $d$ ranging from $\sim$$5{,}000$ to $\sim$$166{,}000$ on the six datasets tested.

\begin{figure}[H]
    \begin{algorithm}[H]
    \begin{algorithmic}[1]
       \STATE {\bfseries Input:} Dataset D, utility $u$ with $\ell_1$ sensitivity $\Delta(u)$, set of outputs $O$, privacy parameter $\eps$
       \STATE $q^* \gets \max_{o \in O} u(D, o)$
       \REPEAT{}
           \STATE $o \sim \mathrm{Uniform}[O]$
           \STATE $p_o \gets \exp\left(\frac{\epsilon}{2\Delta(u)} (u(D, o) - q^*)\right)$
       \UNTIL{$\mathrm{Bernoulli}(p_o)$}
       \STATE {\bfseries Output:} $o$
    \end{algorithmic}
    \caption{Exponential mechanism as rejection sampling}
    \label{alg:exp_mech_as_rejection_sampling}
    \end{algorithm}
\end{figure}

For different data distributions than those displayed by the six test datasets though, there can be an appreciable difference between \joint{} and \pnfjoint{}.  For example, consider a dataset consisting of $d = 4$ items with counts $c_1 = 10$, $c_2 = 5$, and $c_3 = c_4 = 1$.  Then, for $k = 2$ and $\epsilon = 1$, the probability of sampling the item sequence $[1, 2]$ is $\sim$$0.63$ under \joint{}, but $\sim$$0.75$ under \pnfjoint{}.  The same sort of difference can also be observed with larger values of $d$ and $k$.  For instance, consider a dataset consisting of $d = 1000$ items with counts $c_1 = 30$, $c_2 = 15$, and $c_3 = c_4 = \ldots = c_{1000} = 1$.  Then, for $k = 2$ and $\epsilon = 1$, the probability of sampling the item sequence $[1, 2]$ is $\sim$$0.34$ under \joint{}, but $\sim$$0.44$ under \pnfjoint{}.}

\end{document}